\documentclass[a4paper,11pt]{article}
\pdfoutput=1
\usepackage{jheppub}
\renewcommand\acknowledgments{\section*{Acknowledgements}}
\usepackage[utf8]{inputenc}
\usepackage[T1]{fontenc}
\usepackage{amsmath,amssymb,amsthm,mathtools,mleftright,orcidlink}
\usepackage[russian,german,dutch,british]{babel}
\usepackage[final,spacing]{microtype}

\newtheorem{theorem}{Theorem}
\newtheorem{proposition}{Proposition}
\theoremstyle{definition}
\newtheorem{definition}{Definition}
\theoremstyle{remark}
\newtheorem{example}{Example}

\hypersetup {
    pdftitle = {Homotopy Manin Theories: Generalising Third-Way, Yang--Mills and Integrable Sigma Models},
    pdfauthor = {Alexandros Spyridion Arvanitakis, Leron Borsten, Dimitri Kanakaris Decavel, Hyungrok Kim},
    pdfsubject = {hep-th, math-ph},
    pdflang = {en-GB},
    pdfkeywords = {},
}

\title{Homotopy Manin Theories: Generalising Third-Way, Yang--Mills and Integrable Sigma Models}

\author[a,b,c]{Alex S.~\textsc{Arvanitakis}\orcidlink{0000-0002-7844-5574},}
\author[d,e]{Leron \textsc{Borsten}\orcidlink{0000-0001-9008-7725},}
\author[a,b,d]{Dimitri \textsc{Kanakaris}\orcidlink{0009-0001-7716-851X},}
\author[d]{Hyungrok \textsc{Kim}\orcidlink{0000-0001-7909-4510}}

\affiliation[a]{\foreignlanguage{dutch}{Theoretische Natuurkunde, Vrije Universiteit Brussel, Pleinlaan 2, B-1050 Brussel}, Belgium}
\affiliation[b]{International Solvay Institutes, \foreignlanguage{dutch}{Pleinlaan 2, B-1050 Brussel}, Belgium}
\affiliation[c]{\foreignlanguage{croatian}{Institut Ruđer Bošković, Bijenička cesta 54, 10000 Zagreb}, Croatia}
\affiliation[d]{Department of Physics, Astronomy and Mathematics, University of Hertfordshire, Hatfield, Hertfordshire.\ AL10 9AB, United Kingdom}
\affiliation[e]{Blackett Laboratory, Imperial College London,
London SW7 2AZ, United Kingdom}

\emailAdd{alex.s.arvanitakis@irb.hr}
\emailAdd{l.borsten@herts.ac.uk}
\emailAdd{d.kanakaris-decavel@herts.ac.uk}
\emailAdd{h.kim2@herts.ac.uk}
\abstract{
Manin theories  are a class of non-topological deformations of Chern--Simons theories that naturally realise the third-way mechanism and furthermore admit localisation despite not being supersymmetric in the usual sense.
In this paper, we extend this construction to higher dimensions, thereby producing a large class of examples of third-way-type theories.  Furthermore, the construction naturally yields Yang--Baxter integrable deformations of the principal chiral model as well as gravitational models various dimensions.}
\begin{document}
\maketitle

\section{Introduction and summary}
The Alexandrov--Kontsevich--Schwarz--Zaboronsky (AKSZ) sigma model \cite{Alexandrov:1995kv} (reviewed in \cite{Roytenberg:2006qz,Cattaneo:2010re,Kotov:2010wr,Ikeda:2012pv,Kim:2018wvi}) provides a uniform construction of many Schwarz-type topological field theories in any dimension (such as Chern--Simons theory, \(BF\) model, topological Yang--Mills theory \(F\wedge F\), Poisson sigma model, etc.) and provides structural insights in terms of their gauge structure as a symplectic Lie \(n\)-algebroid. They appear in holography \cite{Vassilevich:2013ai,Pulmann:2020omk} and in double field theory \cite{Marotta:2021sia}.

This article extends the construction by deforming the AKSZ sigma models using a Dirac structure (or Manin pair).
This class of (non-topological) theories, which we call  Manin--AKSZ sigma models, provides a uniform construction of such diverse theories as (nonsupersymmetric) Yang--Mills theory, the third-way theories (a subsector of the ABJM model in a Romans background) \cite{Arvanitakis:2015oga,Deger:2021fvv,Deger:2022znj,Kanakaris:2023kfq} (reviewed in \cite{Deger:2021ojb}), the Manin theories \cite{Arvanitakis:2024dbu}, Yang--Baxter integrable deformations \cite{Klimcik:2002zj,Klimcik:2008eq} (reviewed in \cite{yoshida,Hoare:2021dix}) of the principal chiral model, and more.
In particular, this provides a way to construct theories that exhibit the third way mechanism \cite{Arvanitakis:2015oga,Deger:2021fvv,Deger:2022znj,Kanakaris:2023kfq}  using \(L_\infty\)-algebras and \(L_\infty\)-algebroids. Furthermore, many of this class of theories admit an interpretation as a gravitational theory coupled to backgrounds \cite{Borsten:2024pfz,Borsten:2024alh}. For prior   approaches to deforming AKSZ models to generate non-topological theories see, for example, \cite{Barnich:2010sw, Grigoriev:2010ic, Grigoriev:2012xg, Alkalaev:2013hta, Grigoriev:2016wmk, Pulmann:2019vrw, Grigoriev:2020xec} and the references therein.

\section{Geometric structures}
We recall the relevant notions of differential graded geometry in terms of which the Batalin--Vilkovisky formalism is naturally formulated, and define the notions of \emph{admissible subalgebras} and \emph{admissible fibrations} that define the geometry of homotopy Manin theories.

\subsection{Symplectic differential graded manifolds}
AKSZ models are naturally associated to symplectic differential graded manifolds, which we now review, fixing our conventions. For more comprehensive reviews see \cite{Ikeda:2012pv,Jurco:2018sby,Kraft:2022efy}. Our (graded) vector spaces are over the real numbers unless otherwise specified. We use the Koszul sign convention throughout. For a grading indexed by $i\in\mathbb Z$, $V=\bigoplus_i V_i$, suspensions are such that \((V[i])_j=V_{i+j}\). The symbol \(\odot\) denotes graded symmetrisation.

\begin{definition}
An \(L_\infty\)-algebra \((\mathfrak g,\{\mu_i\}_{i=1}^\infty)\) is a graded vector space \(\mathfrak g=\bigoplus_i \mathfrak g^i\) equipped with totally graded-antisymmetric multilinear maps
\begin{equation}
    \mu_i\colon\bigwedge^i\mathfrak g\to\mathfrak g
\end{equation}
of degree \(2-i\) such that
\begin{equation}
    \sum_{\mathclap{\substack{i+j=k\\\sigma\in\operatorname{Sym}(k)}}}\frac{(-1)^{ij}\chi(\sigma)}{i!j!}\mu_{j+1}\mleft(\mu_i(x_{\sigma(1)},\dotsc,x_{\sigma(i)}),x_{\sigma(i+1)},\dotsc,x_{\sigma(k)}\mright)=0.
\end{equation}
\end{definition}
Note that, if \(\mathfrak g\) is an \(L_\infty\)-algebra concentrated in nonnegative degrees, then \(\mathfrak g[1]\) is a differential graded manifold whose body is a single point.

\begin{definition}
A Lie \(n\)-algebroid \((X,Q)\) is a nonnegatively graded manifold \(X\)  concentrated in degrees \(\{-n,1-n,\dotsc,0\}\), together with a homological vector field \(Q\) on \(X\) of degree \(+1\).
\end{definition}

\begin{example}
A Lie 0-algebroid is the same thing as a smooth manifold (and \(Q=0\) necessarily).
\end{example}

\begin{example}
A Lie 1-algebroid \((X,Q)\) is the same thing as a Lie algebroid. Then \(X\) has the structure of a vector bundle \(E[1]\twoheadrightarrow|X|\).
With appropriate local coordinates \((x^i,\theta^a)\), then
\begin{equation}
    Q = \rho_a^i\theta^a\frac\partial{\partial x^i}
    -\frac12f^a{}_{bc}\theta^b\theta^c\frac\partial{\partial\theta^a},
\end{equation}
such that \(\rho\colon E\to\mathrm TX\) defines the anchor map, and \(f\) defines the bracket.
\end{example}

\begin{example}
A Lie \(n\)-algebroid \((X,Q)\) over a point is the same as an \(L_\infty\)-algebra \(\mathfrak g\) concentrated in degrees \(\{1-n,\dotsc,0\}\). Concretely, \(X=\mathfrak g[1]\), such that
\begin{equation}
    \mathcal C^\infty(X)=\bigodot\mathfrak g[1]^*,
\end{equation}
and then \(Q\colon\mathcal C^\infty(X)\to\mathcal C^\infty(X)\) is the Chevalley--Eilenberg differential. Given a basis \(t_a\) of \(\mathfrak g\),
\begin{equation}
    Q = -f^a\frac\partial{\partial t^a}
    - f^a{}_bt^b\frac\partial{\partial t^a}
    - \frac12f^a{}_{bc}t^bt^c\frac\partial{\partial t^a}
    - \dotsb,
\end{equation}
and \(f^a{}_{b_1b_2\dotso b_k}\) defines the structure constants for \(\mu_k\colon\mathfrak g^{\wedge k}\to\mathfrak g\).
\end{example}

\begin{definition}
A symplectic Lie \(n\)-algebroid consists of a Lie \(n\)-algebroid \((X,Q)\) equipped with a nondegenerate closed two-form \(\omega\in\Omega^2(X)\) of degree \(n+2\).\footnote{Differential forms on a graded manifold are bigraded by form degree and the inherent degree of the coordinates.}
\end{definition}
For example, a symplectic 0-algebroid is the same thing as a symplectic manifold. On a smooth manifold \(Y\), the degree-shifted tangent bundle \(\mathrm T[1]Y\) is a Lie (1-)algebroid, and a symplectic structure on it is the same thing as a Poisson structure; a Lie 1-algebroid on the one-point space \(\bullet\) is the same thing as a Lie algebra, and a symplectic strutcure on it is the same thing as an invariant nondegenerate metric.
A symplectic 2-algebroid is the same thing as a Courant algebroid. A symplectic \(n\)-algebroid on \(\bullet\) is the same as (the décalage\footnote{Here, décalage refers to constructing an \(L_\infty[1]\)-algebra from an \(L_\infty\)-algebra; see e.g.\ \cite{zbMATH06126056}.} of) a cyclic \(L_\infty\)-algebra in degrees \(\{1-n,\dotsc,1\}\).

The following proposition is standard.
\begin{proposition}
Given a symplectic \(n\)-algebroid \((X,Q,\omega)\) and a point \(x\in|X|\) in the body \(|X|\) of \(X\), then there exists a canonical cyclic \(L_\infty\)-algebra structure on the tangent space \(\mathrm T_x[-1]X\), where the cyclic structure is given by (the décalage of) \(\omega_x\) and the \(L_\infty\)-algebra structure is given by the Taylor expansion of \(Q\) near \(x\).
\end{proposition}

\subsection{Admissible subalgebras}
The appropriate gauge structures for Manin theories in arbitrary dimensions are \(L_\infty\)-algebras equipped with an isotropic \emph{admissible subalgebra}. The admissible subalgebra defines the gauge subalgebra that is left unbroken by the mass term. This notion reduces to that of a Manin pair for Lie algebras.

\begin{definition}
Let \(\mathfrak d\) be an \(L_\infty\)-algebra with a cyclic pairing of degree \(d-3\) of split signature. An \emph{admissible subalgebra} of \(\mathfrak d\) is a homogeneous Lagrangian subspace \(\mathfrak g\subset\mathfrak d\) that satisfies the following three conditions:
\begin{enumerate}
\item For any \(1\le j\le i\) and \(a_1+\dotsb+a_j\le j-2\), we have
\begin{equation}\label{eq:cond1}
\mu_i(\mathfrak g^{a_1},\dotsc,\mathfrak g^{a_j},\mathfrak d,\dotsc,\mathfrak d)\subset\mathfrak g.
\end{equation}
\item For any \(0\le j\le i-1\) and \(a\ge3-d\) and \(b_1+\dotsb+b_j\le j-2\), we have
\begin{equation}\label{eq:cond2}
\mu_i(\mathfrak g^a,\mathfrak g^{b_1},\dotsc,\mathfrak g^{b_j},\mathfrak d,\dotsc,\mathfrak d)=0.
\end{equation}
\item For any \(a\ge3-d\), \(a+b_1+\dotsb+b_{i-1}\ge2+i-d\), we have
\begin{equation}\label{eq:cond3}
\mu_i(\mathfrak g^a,\mathfrak g^{b_1},\dotsc,\mathfrak g^{b_{i-1}})\subset\mathfrak g.
\end{equation}
\end{enumerate}
\end{definition}

\begin{example}
In a Lie algebra \(\mathfrak g\) (regarded as an \(L_\infty\)-algebra concentrated in degree zero), a linear subspace \(\mathfrak h\subset\mathfrak g\) is an admissible subalgebra if and only if it is a Lie subalgebra.
\end{example}

\begin{example}\label{example:1d_admissible_subalgebra}
Given an vector space \(V\) equipped with a split-signature inner product \(\langle-,-\rangle\),
then \(\mathfrak d\coloneqq V[-1]\) may be regarded as a \(L_\infty\)-algebra with \(\mu_i=0\) for all \(i\) and a cyclic structure of degree \(-2=1-3\). Then an admissible subalgebra is the same as (the degree shift of) a  Lagrangian subalgebra \(L[-1]\subset V[-1]\).
\end{example}

Given a manifold \(\Sigma\) and an \(L_\infty\)-algebra \(\mathfrak d\) the tensor product \(\Omega(M)\otimes\mathfrak d\) carries the structure of an \(L_\infty\)-algebra\footnote{Since the de~Rham algebra \(\Omega(\Sigma)\) carries the structure of a graded-commutative associative algebra.}, with products  given by
\begin{equation}
    \mu_i^{\Omega(M)\otimes\mathfrak d}(\alpha_1\otimes x_1,\dotsc,\alpha_i\otimes x_i)
    \coloneqq\pm(\alpha_1\wedge\dotsb\wedge\alpha_i)\otimes\mu_i^{\mathfrak d}(x_1,\dotsc,x_i)
    + \delta_{i1}\mathrm d\alpha_1\otimes x_i
\end{equation}
for every homogeneous \(\alpha_1,\dotsc,\alpha_i\in\Omega(\Sigma)\) and \(x_1,\dotsc,x_i\in\mathfrak d\),  where \(\pm\) is the Koszul sign corresponding to transpositions of \(\alpha_1,\dotsc,\alpha_i\) and \(x_1,\dotsc,x_i\). Furthermore, it is cyclic if \(\mathfrak d\) is cyclic.
This is, however, too large for homotopy Manin theories, where we wish to kill half of the gauge symmetry (into \(\mathfrak g\)-valued ghosts rather than \(\mathfrak d\)-valued ghosts for an admissible subalgebra \(\mathfrak g\subset\mathfrak d)\)). This breaking of the gauge symmetry is what will ultimately yield propagating degrees of freedom, generalising the construction introduced in \cite{Arvanitakis:2024dbu}. An appropriately smaller \(L_\infty\)-algebra is provided by the following theorem.

\begin{theorem}
Let \(\Sigma\) be a \(d\)-dimensional compact oriented manifold, and \(\mathfrak d\) be a cyclic \(L_\infty\)-algebra concentrated in degrees \(\{2-d,\dotsc,0,1\}\), such that one can construct the \(L_\infty\)-algebra \(\Omega(\Sigma)\otimes\mathfrak g\).
Let \(\mathfrak g\) be an admissible subalgebra of \(\mathfrak d\). Then the following holds.
\begin{enumerate}
\item the homogeneous subspace \(\tilde\Omega(\Sigma;\mathfrak d,\mathfrak g)\subset\Omega(\Sigma)\otimes\mathfrak d\) given by
\begin{equation}
    \tilde\Omega^i(\Sigma;\mathfrak d,\mathfrak g)\coloneqq\begin{cases}
        \bigoplus_{p+q=i}\Omega^p(\Sigma;\mathfrak g^q)& i\le0\\
        \bigoplus_{p+q=i}\Omega^p(\Sigma;\mathfrak d^q)& i\ge1
    \end{cases}
\end{equation}
is an \(L_\infty\)-subalgebra of \(\Omega(\Sigma)\otimes\mathfrak d\).
\item The homogeneous subspace \(\check\Omega^i(\Sigma;\mathfrak g)\subset\tilde\Omega(\Sigma;\mathfrak d,\mathfrak g)\) given by
\begin{equation}
    \check\Omega^i(\Sigma;\mathfrak g)\coloneqq\begin{cases}
        0& i\le2\\
        \bigoplus_{p+q=i}\Omega^p(M;\mathfrak g^q)& i\ge3
    \end{cases}
\end{equation}
is an \(L_\infty\)-ideal of \(\tilde\Omega(\Sigma;\mathfrak d,\mathfrak g)\).
\end{enumerate}
Thus, the subquotient
\begin{equation}
    \hat\Omega(\Sigma;\mathfrak d,\mathfrak g)\coloneqq \tilde\Omega(\Sigma;\mathfrak d,\mathfrak g)/\check\Omega(\Sigma;\mathfrak g)
\end{equation}
exists as an \(L_\infty\)-algebra.
\end{theorem}\begin{proof}[Proof of 1.]
We must show that \(\tilde\Omega(\Sigma;\mathfrak d,\mathfrak g)\subset\Omega(\Sigma)\otimes\mathfrak d\) is closed under \(\mu_i^{\Omega(\Sigma)\otimes\mathfrak d}\) (which we simply write as \(\mu_i\) below).
\begin{itemize}
\item Closure under \(\mu_1\): we must check that, for \(\alpha\otimes x\) with \(\alpha\in\Omega^p(\Sigma)\) and \(x\in\mathfrak g^q\) and \(p+q<1\), that \(\mathrm d\alpha\otimes x+(-1)^p\alpha\otimes\mu^{\mathfrak d}_1(x)\in\tilde\Omega(\Sigma;\mathfrak d,\mathfrak g)\). Now, \(\mathrm d\alpha\otimes x\in\Omega^{p+1}(\Sigma;\mathfrak g^q)\subset\tilde\Omega^{p+q+1}(\Sigma;\mathfrak d,\mathfrak g)\). As for the other term, the argument in the case \(\mu_i\) for general \(i\) applies verbatim.
\item Given the above, it suffices to show that
\begin{multline}
\mu_i\mleft(\Omega^{p_1}(\Sigma;\mathfrak g_{-a_1}),\dotsc,\Omega^{p_j}(\Sigma;\mathfrak g_{-a_j}),\Omega^{q_1}(\Sigma;\mathfrak d_{-b_1}),\dotsc,\Omega^{q_{i-j}}(\Sigma;\mathfrak d_{-b_{i-j}})\mright)\\
\in\Omega(\Sigma)\otimes\mathfrak g
\end{multline}
whenever \(p_1-a_1,\dotsc,p_j-a_j\le0\) and \(q_1-b_1,\dotsc,q_{i-j}-b_{i-j}\ge1\) and
\begin{equation}(2-i)+(p_1-a_1)+\dotsb+(p_j-a_j)+(q_1-b_1)+\dotsb+(q_{i-j}-b_{i-j})\le0.\end{equation}
It thus follows that we can without loss of generality set the form degrees to their minimum values, that is, to ensure
\begin{multline}
\mu_i\mleft(\Omega^0(\Sigma;\mathfrak g_{-a_1})\otimes\dotsc\otimes\Omega^0(\Sigma;\mathfrak g_{-a_j})\otimes\Omega^{1+b_1}(\Sigma;\mathfrak d_{-b_1})\otimes\dotsb\otimes\Omega^{1+b_{i-j}}(\Sigma;\mathfrak d_{-b_{i-j}})\mright)\\
\in\Omega(\Sigma)\otimes\mathfrak g
\end{multline}
whenever
\begin{equation}(2-i)-a_1-\dotsb-a_j+(i-j)\le0.\end{equation}
But this follows from \eqref{eq:cond1}.\qedhere
\end{itemize}
\end{proof}
\begin{proof}[Proof of 2.]
We must check the ideal condition for \(\mu_i^{\Omega(\Sigma)\otimes\mathfrak d}\), that is, to show that the value of \(\mu_i^{\Omega(\Sigma)\otimes\mathfrak d}\) lies in \(\check\Omega(\Sigma;\mathfrak g)\) whenever at least one of its argument belongs to \(\check\Omega(\Sigma;\mathfrak g)\) and the rest (if any) belong to \(\tilde\Omega(\Sigma;\mathfrak d,\mathfrak g)\). As before, we simply write \(\mu_i\) for \(\mu_i^{\Omega(\Sigma)\otimes\mathfrak d}\).

The operator \(\mu_1(\alpha\otimes x)=\mathrm d\alpha\otimes x+(-1)^{|\alpha|}\alpha\otimes\mu_1^{\mathfrak d}(x)\) contains two terms, among which 
the first term always increases form degree and hence can pose no problem.
Hence it suffices to ensure that
\begin{multline}
    \mu_i(\Omega^p(\Sigma;\mathfrak g^{-a}),
    \Omega^{q_1}(\Sigma;\mathfrak g^{-b_1}),
    \dotsc,
    \Omega^{q_j}(\Sigma;\mathfrak g^{-b_j}),
    \dotsc,
    \Omega^{r_1}(\Sigma;\mathfrak d^{-c_1}),\dotsc,\Omega^{r_{i-j-1}}(\Sigma;\mathfrak d^{-x_{i-1}}))\\
    \in\check\Omega(\Sigma;\mathfrak g)\label{eq:to_be_shown}
\end{multline}
whenever \(p-a\ge3\)
and \(r_1-c_1,\dotsc,r_{i-j-1}-c_{i-j-1}\ge1\).
There are two ways in which \eqref{eq:to_be_shown} could fail: (a) the total degree might dip to \(<3\) so that we fall out of \(\check\Omega(\Sigma;\mathfrak d,\mathfrak g)\); (b) the \(\mu_i^{\mathfrak d}\) might make me fall out of \(\check\Omega(\Sigma;\mathfrak d,\mathfrak g)\) in that we end up in \(\tilde\Omega(\Sigma;\mathfrak d,\mathfrak g)\setminus\Omega(\Sigma;\mathfrak g)\).
To avoid (a), it suffices without loss of generality to consider the case when all form degrees are minimised, that is, to ensure
\begin{multline}
    \mu_i(\Omega^{3-a}(\Sigma;\mathfrak g^{-a}),
    \Omega^0(\Sigma;\mathfrak g^{-b_1}),\dotsc,
    \Omega^0(\Sigma;\mathfrak g^{-b_j}),
    \Omega^{c_1+1}(\Sigma;\mathfrak d^{-c_1}),\dotsc,\Omega^{c_{i-j-1}+1}(\Sigma;\mathfrak d^{-c_{i-1}}))\\
    =0
\end{multline}
whenever
\begin{equation}
    (2-i)+3-b_1-\dotsb-b_j+(i-j-1)\le2.
\end{equation}
But this follows from \eqref{eq:cond2}.
To avoid (b), we need to ensure that
\begin{multline}
    \mu_i(\Omega^p(\Sigma;\mathfrak g^{-a}),
    \Omega^{q_1}(\Sigma;\mathfrak g^{-b_1}),
    \dotsc,
    \Omega^{q_j}(\Sigma;\mathfrak g^{-b_j}),
    \dotsc,
    \Omega^{r_1}(\Sigma;\mathfrak d^{-c_1}),\dotsc,\Omega^{r_{i-j-1}}(\Sigma;\mathfrak d^{-x_{i-1}}))\\
    \in\Omega(\Sigma;\mathfrak g)
\end{multline}
whenever \(p-a\ge3\) and \(r_1-c_1,\dotsc,r_{i-j-1}-c_{i-j-1}\ge1\) and
\begin{equation}
    (2-i)+(p-a)+(q_1-b_1)+\dotsb+(q_j-b_j)+(r_1-c_1)+\dotsb+(r_{i-j-1}-c_{i-j-1})\ge3
\end{equation}
and
\begin{equation}
    p+q_1+\dotsb+q_j+r_1+\dotsb+r_{i-j-1}\le d.
\end{equation}
That is, given \((-a,-b_1,\dotsc,-b_j,-c_1,\dotsc,-c_{i-j-1})\), if we can choose \((p,q,r)\) so that the above inequality is satisfied, then the corresponding \(\mu_i^{\mathfrak d}\) has to take values in \(\mathfrak g\).
\begin{equation}
\begin{aligned}
    \MoveEqLeft 1+i+a+(b_1-q_1)+\dotsb+(b_j-q_j)+(c_1-r_1)+\dotsb+(c_{i-j-1}-r_{i-j-1})\\
    &\le p\\&\le d-q_1-\dotsb-q_j-r_1-\dotsb-r_{i-j-1}.
\end{aligned}
\end{equation}
Such a \(p\) exists iff
\begin{equation}
    a+b_1+\dotsb+b_j+c_1+\dotsb+c_{i-j-1}\le d-1-i.
\end{equation}
But now we see that the condition \eqref{eq:cond3} is the condition that we need.
\end{proof}

\begin{definition}
    Given a multiset of nonnegative integers \(S\), property (A) holds iff there is at most one element of \(S\) is positive:
    \begin{equation}
        \text{(A)}\iff \#\{s\in S\mid |s|>0\} \le 1.
    \end{equation}
    (Here \(\#\) denotes the cardinality of a multiset.)
    Property (B) holds iff all but two elements of \(S\) are zero, and the two remaining elements are equal:
    \begin{equation}
        \text{(B)}\iff \exists p\in\mathbb N\colon S = \{0,\dotsc,0\}\sqcup\{p,p\}.
    \end{equation}
    (Properties (A) and (B) hold simultaneously iff \(S\) only contains zeros.)
\end{definition}

\begin{definition}
A \emph{Hodge structure} on an admissible subalgebra \(\mathfrak g\) of a cyclic \(L_\infty\)-algebra \(\mathfrak d\) (where \(\mathfrak d\) carries cyclic structure of degree \(3-d\)) consists of nondegenerate symmetric (\emph{not} graded-symmetric) bilinear pairings
\begin{equation}
    \varkappa^{(i)}\colon(\mathfrak d^i/\mathfrak g^i)\otimes(\mathfrak d^i/\mathfrak g^i)\to\mathbb R
\end{equation}
for each degree \(i\), which we can regard as linear maps
\begin{equation}
    M^{(i)}\colon\mathfrak d^i\to\mathfrak d^{d-3-i}
\end{equation}
with
\begin{equation}\operatorname{im}M^{(i)}=\mathfrak g^i=\ker M^{(i)}\qquad(i\le0),\end{equation}
and that is cyclic in that
\begin{equation}\label{eq:hodge_cyclicity}
    \langle Mx,y\rangle = \langle x,My\rangle
\end{equation}
for any \(x,y\in\mathfrak d\), such that given \(a_1,\dotsc,a_i\in\{2-d,\dotsc,1\}\) and \(p_1,\dotsc,p_i\in\{0,\dotsc,d\}\),
\begin{itemize}
\item if neither (A) nor (B) hold for the multiset \(\{p_1,\dotsc,p_i\}\), then
\begin{align}\label{eq:generic_case}
X^{a_1,p_1;\dotsb;a_i,p_i}&\in V&
Y_k^{a_1,p_1;\dotsb;a_i,p_i}&\in V
\end{align}
for all \(k\in\{1,\dotsc,i\}\);
\item if only (A) holds, then
\begin{equation}\label{eq:cond_(A)_case}
    X^{a_1,p_1;\dotsb;a_i,p_i}+Y_k^{a_1,p_1;\dotsb;a_i,p_i}\in V
\end{equation}
for all \(k\in\{1,\dotsc,i\}\) and \(p_k\ge1\);
\item if only (B) holds, then
\begin{equation}\label{eq:cond_(B)_case}
    Y_k^{a_1,p_1;\dotsb;a_i,p_i}+(-1)^{p_k(d-p_k)}Y_l^{a_1,p_1;\dotsb;a_i,p_i}\in V
\end{equation}
for all \(k,l\in\{1,\dotsc,i\}\) with \(p_k=p_l\ge1\);
\item if both (A) and (B) hold (i.e.\ \(p_1=\dotsb=p_i=0\)) , then
\begin{equation}\label{eq:cond_(AB)_case}
X^{a_1,p_1;\dotsb;a_i,p_i}+\sum_{k=1}^iY_k^{a_1,p_1;\dotsb;a_i,p_i}\in V.
\end{equation}
\end{itemize}
In the above,
\begin{multline}
    X^{a_1,p_1;\dotsb;a_i,p_i}
    \coloneqq\\\begin{cases}
        M(\mu_i(\tilde{\mathfrak g}^{a_1},\dotsc,\tilde{\mathfrak g}^{a_i})) &\text{if \(p_1+\dotsb+p_i\le d \text{ and }a_1+p_1+\dotsb+a_i+p_i=i-1\)} \\
        0 &\text{otherwise}
    \end{cases}
\end{multline}
and
\begin{multline}
    Y_k^{a_1,p_1;\dotsb;a_i,p_i}
    \coloneqq\\\begin{cases}
        \mu_i(\tilde{\mathfrak g}^{a_1},\dotsc,M(\tilde{\mathfrak g}^{a_k})\dotsc,\tilde{\mathfrak g}^{a_i}) & \text{if \(p_1+\dotsb+p_i\le2p_k \text{ and }a_k+p_k=1\)}\\
        0 &\text{otherwise}
    \end{cases}
\end{multline}
where \(M\) is applied to the \(k\)th argument (\(k\in\{1,\dotsc,i\}\)), and in which
\begin{equation}
    \tilde{\mathfrak g}^{a_k}\coloneqq\begin{cases}
        \mathfrak g^{a_k} & a_k+p_k\le 0\\
        \mathfrak d^{a_k} & a_k+p_k\ge1
    \end{cases}
\end{equation}
and
\begin{equation}
    V =\begin{cases}
    \mathfrak g & \text{if \(a_1+p_1+\dotsb+a_i+p_i\ge i\)} \\
0 & \text{otherwise}.
\end{cases}
\end{equation}
\end{definition}
\begin{example}
Let \(\mathfrak d\) be a cyclic Lie algebra with a Lagrangian subalgebra \(\mathfrak g\subset\mathfrak d\). Then a Hodge structure on \((\mathfrak d,\mathfrak g)\) is a linear map \(M\colon\mathfrak d\to\mathfrak d\) with \(\operatorname{im}M=\mathfrak g=\ker M\) such that
\begin{equation}
    \langle Mx,y\rangle = \langle x,My\rangle 
\end{equation}
for any \(x,y\in\mathfrak d\) and
\begin{align}
    M[x,y]=[x,My]
\end{align}
for any \(x\in\mathfrak g\) and \(y\in\mathfrak d\).
\end{example}

\begin{example}
Let \(\mathfrak g\) be an \(L_\infty\)-algebra with cyclic pairing of degree \(3-n\). Then
\begin{equation}\mathfrak d=\mathfrak g\oplus\mathfrak g^*\end{equation}
forms a Manin pair. A Hodge structure consists of suitable (anti-)symmetric pairings on each of \(\mathfrak g_i\).
\end{example}

\begin{theorem}
Let \(\Sigma\) be a \(d\)-dimensional compact oriented Riemannian manifold, and \(\mathfrak d\) be a cyclic \(L_\infty\)-algebra concentrated in degrees \(\{2-d,\dotsc,-1,1\}\) with an admissible subalgebra \(\mathfrak g\). Let \(M\) be a Hodge structure on \((\mathfrak d,\mathfrak g)\). Then the \(L_\infty\)-algebra structure \(\mu_i^{\hat\Omega^i(\Sigma;\mathfrak d,\mathfrak g)}\) on \(\hat\Omega(\Sigma;\mathfrak d,\mathfrak g)\) admits a one-parameter deformation
\begin{align}
    \mu_1^M&\coloneqq \mu^{\hat\Omega(\Sigma;\mathfrak d,\mathfrak g)}_1+t\nu&
    \mu_i^M&\coloneqq\mu^{\hat\Omega(\Sigma;\mathfrak d,\mathfrak g)}_i\qquad\forall i\ge2
\end{align}
where \(t\) is an arbitrary real number and
\begin{equation}
    \nu(\alpha\otimes x)\coloneqq\begin{cases}
        \star\alpha\otimes M(x) & \text{if \(|x|+|\alpha|=1\)}\\
        0 & \text{otherwise}
    \end{cases}
\end{equation}
for any  \(\alpha\in\Omega(\Sigma)\) and \(x\in\mathfrak d\), where \(\star\) is the Hodge star on \(\Omega(M)\). Furthermore, this deformation is cyclic.\footnote{It is possible to generalise to topological deformations that do not require the Hodge star \cite{Borsten:2024alh}. We leave this analysis for future work.}
\end{theorem}\begin{proof}
We first show that the \(L_\infty\)-algebra homotopy Jacobi identities hold for \(\mu_i^M\). It is convenient to decompose the \(L_\infty\)-algebra homotopy Jacobi identities according to the power of the formal parameter \(t\), so that we have the \(\mathcal O(t^0)\), \(\mathcal O(t)\), and \(\mathcal O(t^2)\) components, which must each vanish separately. 
\begin{itemize}
\item The \(\mathcal O(t^0)\) component of the homotopy Jacobi identities for \(\mu_i^M\) is the original homotopy Jacobi identities of \(\tilde\Omega^i(\Sigma;\mathfrak d,\mathfrak g)\), so there is nothing to check.
\item The \(\mathcal O(t^2)\) component of the homotopy Jacobi identities for \(\mu_i^M\) vanishes since \(\nu\circ\nu=0\) for degree reasons: \(\nu\) is only nonzero on elements of degree one.
\item The \(\mathcal O(t)\) component of the homotopy Jacobi identities for \(\mu_i^M\) is
\begin{equation}\label{eq:O(t)_component_identity}
    \nu(\mu^{\hat\Omega^i(\Sigma;\mathfrak d,\mathfrak g)}_i(x_1,\dotsc,x_i))
    +
    \sum_{j=1}^i\pm\mu_i^{\hat\Omega^i(\Sigma;\mathfrak d,\mathfrak g)}(x_1,\dotsc,\nu(x_j),\dotsc,x_i)
    =0.
\end{equation}
The operator \(\mu_i^{\hat\Omega(\Sigma;\mathfrak d,\mathfrak g)}\) has two kinds of terms: one (only present on \(\mu_1^{\hat\Omega(\Sigma;\mathfrak d,\mathfrak g)}\)) involving the exterior derivative of forms, and the other involving \(\mu_i^{\mathfrak d}\) and wedge products of differential forms.

Let us first check the exterior derivative term on \(\mu_1^{\hat\Omega(\Sigma;\mathfrak d,\mathfrak g)}\). Given \(\alpha\in\Omega^\bullet(X)\) and \(x\in\mathfrak d\), the \(L_\infty\)-algebra identity is
\begin{equation}
    \left[|\alpha|+|x|=0\right]\star\mathrm d\alpha\otimes M(x)
    +
    \left[|\alpha|+|x|=1\right]\mathrm d\star\alpha\otimes M(x)
    \in V,
\end{equation}
where \([\dotsb]\) is the Iverson bracket \cite{Knuth} (\(1\) if the enclosed statement is true,  \(0\) otherwise). When \(|\alpha|+|x|=0\), we need \(\alpha\in\mathfrak g\), and the term vanishes since \(M(\mathfrak g)=0\). When \(|\alpha|+|x|=1\), then \(M(x)\mathrm d\star\alpha\) carries degree 3, and since \(M(x)\in\mathfrak g\), the result belongs to \(V\). (As a special case: if \(|x|=1\), then the first Iverson bracket can never be true; the second Iverson bracket is nonzero only if \(|x|=0\), but then the second term vanishes for form-degree reasons.)

Let us consider the other terms in \(\mu_i^{\hat\Omega(\Sigma;\mathfrak d,\mathfrak g)}\), that is, those which are \(\mu_i^{\mathfrak d}\) with the form legs wedged together. In \eqref{eq:O(t)_component_identity}, applied to \(i\) arguments of form degrees \(p_1,\dotsc,p_i\) and internal degrees \(a_1,\dotsc,a_i\), one sees that since Hodge stars don't distribute across wedge products, the terms can't cancel each other unless condition (A) or condition (B) holds. When neither hold, then each term must vanish individually, corresponding to \eqref{eq:generic_case}. When condition (A) holds (all but one are 0-forms), we have the identity
\begin{equation}
    \star(\alpha_0\wedge\alpha_1\wedge\dotsb\wedge\alpha_i)
    =
    \alpha_0\wedge\alpha_1\wedge\dotsb(\star\alpha_k)\wedge\dotsb\wedge\alpha_i
\end{equation}
where all the \(\alpha\)s are zero-forms except for \(\alpha_k\), so that \eqref{eq:cond_(A)_case} suffices. When condition (B) holds (all but two are 0-forms, and the two have the same form degree), we have the identity
\begin{equation}
    \alpha_0\wedge\dotsb\wedge(\star\alpha_k)\wedge\dotsb\wedge\alpha_i
    =
    (-1)^{|\alpha_k|(d-|\alpha_k|)}
    \alpha_0\wedge\dotsb\wedge(\star\alpha_l)\wedge\dotsb\wedge\alpha_i
\end{equation}
where \(\alpha_k\) and \(\alpha_l\) are the two forms with positive degree, and then \eqref{eq:cond_(B)_case} suffices. When all forms are 0-forms, then we can relax the condition to \eqref{eq:cond_(AB)_case}.
\end{itemize}

Finally, cyclicity of \(\mu^M_1\) follows from \eqref{eq:hodge_cyclicity}.
\end{proof}

\subsection{Admissible fibrations on symplectic dg-manifolds}
We now formulate a notion analogous to Manin pairs in a Lie algebra, which is necessary to consider homotopy Manin theories with nonlinear target spaces.

\begin{definition}
An \emph{admissible fibration} consists of a graded vector bundle \(p\colon X\twoheadrightarrow Y\) on a graded manifold \(Y\) together with a homological vector field \(Q\) and a symplectic form \(\omega\) on \(X\) and an Ehresmann connection
\begin{equation}
    \mathrm TX = \mathrm V_p \oplus H
\end{equation}
(where \(\mathrm V_p\) is the vertical bundle) such that:
\begin{itemize}
\item \(\omega\) is bilinear on each fibre of \(p\)
\item \(Q\) is a finite sum of homogeneous components with respect to the linear structure on the fibres of \(p\),
\item for any every \(x\in X\), the fibre
\begin{equation}
    H_x\subset\mathrm T_xX
\end{equation}
defines an isotropic subspace of the cyclic \(L_\infty\)-algebra \(\mathrm T_xX\) that is also an admissible subalgebra.
\end{itemize}
\end{definition}
Note that this notion is similar to, but differs from, the notions of Dirac structures in \cite{Arvanitakis:2021wkt}, \(\Lambda\)-structures in \cite{Severa:2001tze}, or homotopy Manin pairs defined in \cite[Def.~34]{2007LMaPh..81...19K}. Our notion is adapted to the current context of homotopy Manin theories.

\begin{definition}
A \emph{Hodge structure} on an admissible fibration \((p\colon X\to Y,H)\) is a family of Hodge structures defined for the family of admissible subalgebras \(H_x\subset \mathrm T_xX\) for each \(x\in X\) that is smooth with respect to \(X\) and constant along the fibre directions.
\end{definition}

\section{Homotopy Manin sigma models}
Using the algebraic and geometric structures from the previous section, we now deform the topological AKSZ sigma models to construct non-topological homotopy Manin sigma models.

Let us recall the construction of an AKSZ sigma model. Given a spacetime (or worldvolume) \(\Sigma\) and a symplectic \(n\)-algebroid \((X,Q,\omega)\), then the space
\begin{equation}
    \mathcal M\coloneqq\mathcal C^\infty(\mathrm T[1]\Sigma,X)
\end{equation}
of graded-smooth maps carries canonically the structure of a symplectic (Fréchet) dg-manifold \((\mathcal M,Q_{\mathcal M},\omega_{\mathcal M})\) with a symplectic form \(\omega_{\mathcal M}\) of degree \(-1\), whose coordinates are concentrated in degrees \(-n,\dotsc,0\). This may be regarded as a BV-extended configuration space satisfying the classical master equation.

Let \(M\) be a Hodge structure on an admissible fibration \(L\) of \(X\). Let us deform the differential \(Q_{\mathcal M}\) by
\begin{equation}Q'_{\mathcal M} = Q_{\mathcal M} + \star M t^a\frac\partial{\partial t^a}.\end{equation}
where \(t_a\) are the coordinates of degree \(0\) (the fields) and \(t^a\) are the corresponding coordinates of degree \(1\) (the antifields) with a DeWitt index \(a\). This vector field is not nilpotent on \(\mathcal M\).
However, let \(\tilde{\mathcal M}\) be the subspace of \(\mathcal M\) consisting only of the following coordinates:
\begin{itemize}
    \item Any coordinates of degrees \(0\) or \(1\)
    \item Any coordinates of degrees other than \(0\) or \(1\) that are constant along the fibres of the admissible fibration
\end{itemize}
Then \(Q'_{\mathcal M}\) restricted to \(\tilde{\mathcal M}\) is nilpotent, so that \((\tilde{\mathcal M},Q'_{\mathcal M},\omega|_{\tilde{\mathcal M}})\) forms a symplectic dg-manifold.
The \emph{homotopy Manin theory} associated to the data \((X,L,M)\) is the classical field theory whose BV formulation is given by the above symplectic dg-manifold.

\begin{example}
Consider the case where the body of \(X=\mathfrak d[1]\) is a single point.
Then we are simply dealing with a cyclic \(L_\infty\)-algebra \(\mathfrak d\) and an admissible subalgebra \(\mathfrak g\subset\mathfrak d\). Then the \(L_\infty\)-algebra associated to the homotopy Manin theory has the underlying graded vecor space
\begin{equation}
    \mathfrak G=\bigoplus_{p,q}\Omega^p(\Sigma)\otimes\mathfrak d_p^q\subset\Omega(\Sigma)\otimes\mathfrak d,
\end{equation}
where
\begin{equation}
    \mathfrak d_p^q\coloneqq\begin{cases}
        \mathfrak d^q/\mathfrak g^q & 2<p+q \\
        \mathfrak d^q & 1\le p+q\le2 \\
        \mathfrak g^q & p+q <1.
    \end{cases}
\end{equation}
The \(L_\infty\)-algebra structure \(\mu_i\) on \(\mathfrak G\) is given by restriction of the \(L_\infty\)-algebra structure \(\mu^{\Omega^\bullet(M)\otimes\mathfrak d}_i\) on \(\Omega^\bullet(M)\otimes\mathfrak d\) to the above subspace, except that \(\mu_1\) has an extra term:
\begin{equation}
\begin{aligned}
    \mu_i &=\mu^{\Omega^\bullet(M)\otimes\mathfrak d}_i|_{\mathfrak G}\qquad(i>1)\\
    \mu_1(x)&=\begin{cases}
    \mu_1^{\Omega^\bullet(M)\otimes\mathfrak d}(x) + M\star x&|x|=1\\
    \mu_1^{\Omega^\bullet(M)\otimes\mathfrak d}(x) & |x|\ne1.
    \end{cases}
\end{aligned}
\end{equation}
\end{example}

\subsection{Action of homotopy Manin models in local coordinates}
Given that the homotopy Manin theory is given by a symplectic dg-manifold, it admits an action formulation. (As such, the class of homotopy Manin theories described here cannot suffer from the kind of failure of unitarity described in \cite{Chakrabarti:2024crx}.)

We can write down the corresponding classical homotopy Maurer--Cartan (hMC) action for the (deformed) cyclic \(L_\infty\)-algebra \(\hat\Omega(\Sigma;\mathfrak d,\mathfrak g),\mu_i^M\):
\begin{equation}
\begin{aligned}
    S&=\int_\Sigma\left\langle\mathbb A,\sum_i\frac1{(i+1)!}\mu_i^M(\mathbb A,\dotsc,\mathbb A)\right\rangle\\
    &=
    \int_\Sigma
    \left(
    \frac12\left\langle\mathbb A,\mathrm\mathrm d\mathbb A\right\rangle
    +
    \frac1{(i+1)!}\left\langle\mathbb A,\sum_i\mu_i^{\mathfrak d}(\mathbb A,\dotsc,\mathbb A)\right\rangle+\frac12\langle\mathbb A,\star M\mathbb A\rangle\right)
\end{aligned}
\end{equation}
restricted to  the space of ordinary fields (rather than antifields, ghosts, or ghost antifields), which is given by  \(\mathbb A\in \bigoplus_i\Omega^{1+i}(\Sigma;\mathfrak d_{-i})\).

Picking a splitting for convenience, this is
\begin{equation}
    \bigoplus_i\Omega^{1+i}(\Sigma;\mathfrak g_{-i})\oplus\Omega^{1+i}(X;(\mathfrak g_i)^*).
\end{equation}
Call the first component \(A\) (the dynamical field) and the second component \(\tilde A\) (the field made auxiliary by the mass term).
Then the action is
\begin{equation}
    \int_\Sigma\operatorname{hMC}(A\oplus\tilde A)+\frac12\langle\tilde A\wedge\star M\tilde A\rangle.
\end{equation}

\subsection{Energy positivity}

We briefly examine the Hamiltonian formulation of an arbitrary homotopy Manin gauge theory in order to demonstrate that the Hamiltonian (\emph{sans} constraints) is bounded from below. In other words these gauge theories have positive energy.

The calculation is essentially the same as in \cite{Arvanitakis:2024dbu}.
First, we perturbatively expand the symplectic \(L_\infty\)-algebroid into a cyclic \(L_\infty\)-algebra \(\mathfrak d=\bigoplus_{i=2-d}^1\mathfrak d_i\), now concentrated in degrees \(\{2-d,\dotsc,1\}\) (where \(d\) is the dimension of spacetime \(\Sigma\)).
Let us write the homotopy Manin action as
\begin{equation}
    \int_\Sigma\Theta(\mathbb A)+\frac12\big\langle \mathbb A\, ,\mathrm{d} \mathbb A \big\rangle +\frac12\big\langle\mathbb A,*M\mathbb A\big\rangle\,,
\end{equation}
where
\begin{enumerate}
    \item $\mathbb A$ is a polyform whose $p$-form component is valued in \(\mathfrak d_{1-p}\);
    \item $\Theta(\mathbb A)$ is the contribution to the action from the target-space $L_\infty$ brackets, and thus involves no derivatives of $\mathbb A$; and
    \item we have set all antifields, ghosts, etc.~to zero. 
\end{enumerate}
For the Hamiltonian formulation, we split the field \(\mathbb A\) into time and space components:
\begin{equation}
\mathbb A=\mathrm{d} x^0 \mathbb A_0 +\mathbb A_\Xi\,,\qquad \iota_{\partial_0}\mathbb A_0=\iota_{\partial_0}\mathbb A_\Xi=0\,,\qquad \mathrm{d}\mathbb A=\mathrm{d} x^0 \wedge\dot {\mathbb A}_\Xi + \mathrm{d}_\Xi \mathbb A_\Xi -\mathrm{d} x^0 \wedge\mathrm{d}_\Xi \mathbb A_0,
\end{equation}
where spacetime $\Sigma$ is also split into time and space as $\Sigma=\mathbb R\times \Xi$ and $x^0$ runs along $\mathbb R$, and \(\iota\) denotes the interior product, and \(\dot{}\) denotes derivative along \(x^0\). (In fact everything below is also valid locally up to integration by parts.)  First, we evaluate the term involving derivatives:
\begin{equation}
\int_\Sigma \eta\Big( \mathbb A \mathrm{d} \mathbb A\Big)=\int_\Sigma \eta\mleft( 2\mathrm{d} x^0\mathbb A_0 \mathrm{d}_\Xi \mathbb A_\Xi + \mathbb A_\Xi \mathrm{d} x^0\dot{\mathbb A}_\Xi\mright),
\end{equation}
from which we see that $\mathbb A_\Sigma$ together are the positions and momenta, while $\mathbb A_0$ appears linearly. In fact $\mathbb A_0$ will appear linearly also in $\Theta(\mathbb A)$ (because it must be linear in $\mathrm{d} x^0$) but it appears quadratically in the mass term:
\begin{equation}
\begin{aligned}
    \Theta(\mathbb A)&=\mathrm{d} x^0\wedge\mathbb A_0\wedge\mathbb G',\\\frac{1}{2}\eta\Big( (M\mathbb A) \wedge{\star \mathbb A}\Big)&=\frac{1}{2}\eta\mleft( \mathrm{d} x^0\wedge(M\mathbb A_0)\wedge{\star (\mathrm{d} x^0 \wedge\mathbb A_0)} + \frac{1}{2}(M\mathbb A_\Xi)\wedge {\star \mathbb A_\Xi} \mright.\\&\hspace{2cm}\mleft.+ \underbrace{\mathrm{d} x^0\wedge(M\mathbb A_0)\wedge {\star (\mathbb A_\Xi)}}_{=0 \text{ if }g^{00}\neq 0}\mright).
\end{aligned}\end{equation}
Here $\mathbb G'$ (which depends on $\mathbb A_\Xi$) is part of the Gauss law constraint for the AKSZ sigma model and is defined by the above formula; the last term in the second equation may be seen to vanish by linear algebra --- if the $\mathrm{d} x^\mu$ are e.g.\ orthonormal, $\star\mathbb A_\Sigma$ will contain $\mathrm{d} x^0$ and thus $\mathrm{d} x^0 \star{\mathbb A}_\Sigma=0$.

Let us now  introduce another split:
\begin{equation}
\frak d=\ker M+\frak e\,;\qquad \mathbb A_0=\underbrace{\mathbb L_0}_{\mathclap{\in \ker M}} +\underbrace{\mathbb E_0}_{\mathclap{\in\frak e}},
\end{equation}
where $\frak e$ is any Lagrangian complement of $\frak g$ in $\frak d$, which may be identified with $\frak g^\star$.
(We do not split $\mathbb A_\Xi$ this way.) Then the action takes the form
\begin{equation}
S=\int_\Sigma\Big\langle\mathbb A_\Xi \wedge\mathrm{d} x^0\wedge\dot{\mathbb A}_\Xi + \mathrm{d} x^0(\mathbb L_0+\mathbb E_0)\wedge \mathbb G +\frac{1}{2}(M\mathbb A_\Xi)\wedge {\star \mathbb A_\Xi} +\frac{1}{2}\mathrm{d} x^0\wedge(M\mathbb E_0)\wedge {\star (\mathrm{d} x^0 \wedge\mathbb E_0)}\Big\rangle
\end{equation}
with $\mathbb G\coloneqq\mathbb G'+2 \mathrm{d}_\Xi \wedge\mathbb A_\Xi$ being the AKSZ model Gauss law constraint; it depends only on $\mathbb A_\Xi$.

The  fields to be varied in the action are now $\mathbb A_\Sigma$, $\mathbb L_0$ and $\mathbb E_0$. Of these, $\mathbb L_0$ is a Lagrange multiplier enforcing certain components of  the AKSZ sigma model Gauss law. Since \(\operatorname{im}M=\mathfrak g=\ker M\), we see that \(\langle M\bullet,\bullet\rangle|_{\frak e}\) is nondegenerate, and $\mathbb E_0$ may be integrated out algebraically. To do this requires use of projection/rejection identities to handle the Hodge star of a wedge, which give
\begin{equation}
\langle\mathrm{d} x^0(M\mathbb E_0) ,{\star (\mathrm{d} x^0 \mathbb E_0)}\rangle=g^{00}\langle M,\mathbb E_0 {\star}\mathbb E_0\rangle\,.
\end{equation}
Completing the square then gives exactly the same sort of term as in \cite{Arvanitakis:2024dbu}; this is a perfect square with a $-g^{00}$ prefactor which will be non-negative in Lorentzian signature (so $g^{00}<0$) assuming also that  $\langle M\bullet,\bullet\rangle|_{\frak e}$ is positive-semidefinite. In that case, the term $\langle\tfrac{1}{2}(M\mathbb A_\Sigma) {\star (\mathbb A_\Sigma)}\rangle$ is non-negative as well; together these terms form the Hamiltonian, which is manifestly non-negative.

\section{Examples}
Let us discuss particular examples of the class of homotopy Manin theories constructed above.

\subsection{General classes of examples}
First, we list those examples that can be defined in an arbitrary number of spacetime (or worldvolume) dimensions.

\begin{example}[First-order Yang--Mills theory]
Let \(\mathfrak g\) be a Lie algebra. Then \((\mathfrak d,\mathfrak g)\) is a Manin pair, where
\begin{equation}
    \mathfrak d=\mathfrak g^*[n-3]\oplus\mathfrak g
\end{equation}
carries Lie bracket
\begin{align}
    [x,y]_{\mathfrak d}&=[x,y]_{\mathfrak g} & 
    [x,\tilde x]_{\mathfrak d}&=\operatorname{coad}_x(\tilde x)&
    [\tilde x,\tilde y]&=0
\end{align}
for \(x,y\in\mathfrak g\) and \(\tilde x,\tilde y\in\mathfrak g^*[n-3]\), where \(\operatorname{coad}\) is the coadjoint representation.

Suppose that \(\kappa\colon\mathfrak g\otimes\mathfrak g\to\mathbb R\) is an invariant metric on \(\mathfrak g\), which induces the musical isomorphism
\begin{align}
    \varkappa^\sharp\colon\mathfrak g&\to\mathfrak g^*\\
    x&\mapsto\kappa(x,-).
\end{align}
Then 
\begin{align}
    M\colon\mathfrak d&\to\mathfrak d \\
    (x\oplus\tilde x[n-3])&\mapsto(\varkappa^\sharp)^{-1}(\tilde x)
\end{align}
(for any \(x\in\mathfrak g\), \(\tilde x\in\mathfrak g^*\)) is a Hodge structure on \((\mathfrak d,\mathfrak g)\).

Given a \(d\)-dimensional oriented (pseudo-)Riemannian manifold \(\Sigma\), then
\[
    \Omega(\Sigma)\otimes\mathfrak d
    =
    \Omega(\Sigma)\otimes\mathfrak g
    \oplus
    (\Omega(\Sigma)\otimes\mathfrak g^*)[d-3].
\]
The subquotient \(\hat\Omega(\Sigma;\mathfrak d,\mathfrak g)\) may be identified with the subspace of \(\Omega(\Sigma;\mathfrak d)\) that excludes the following:
\begin{itemize}
\item elements of the form \(\alpha\otimes x\) where \(\alpha\in\Omega^p\) and \(x\in\mathfrak g^*[d-3]\) and \(p-d+3\le0\) (such elements do not lie in \(\tilde\Omega(\Sigma;\mathfrak d,\mathfrak g)\))
\item elements of the form \(\alpha\otimes x\) where \(\alpha\in\Omega^p\) and \(x\in\mathfrak g\) and \(p\ge3\) (such elements lie in \(\check\Omega(\Sigma;\mathfrak g)\)).
\end{itemize}
That is, we may identify
\begin{multline}
    \hat\Omega(\Sigma;\mathfrak d,\mathfrak g)\cong
    \underbrace{\underset c{\Omega^0(\Sigma;\mathfrak g)}}_{\text{ghost}}
    \oplus
    \underbrace{\underset A{\Omega^1(\Sigma;\mathfrak g)}
    \oplus
    \underset B{\Omega^{d-2}(\Sigma;\mathfrak g^*[d-3])}}_{\text{fields}}\\
    \oplus
    \underbrace{\underset{A^+}{\Omega^{d-1}(\Sigma;\mathfrak g^*[d-3])}
    \oplus
    \underset{B^+}{\Omega^2(\Sigma;\mathfrak g)}}_{\text{antifields}}
    \oplus
    \underbrace{\underset{c^+}{\Omega^d(\Sigma;\mathfrak g^*[d-3])}}_{\text{ghost antifield}},
\end{multline}
which agrees with the field content for first-order Yang--Mills theory. The associated homotopy Maurer--Cartan action is
\begin{equation}
    S_{\hat\Omega(\Sigma;\mathfrak d,\mathfrak g)}=\int_\Sigma B\wedge\left(\mathrm dA+\frac12[A,A]\right)
    + c(\mathrm dA^++A\wedge A^++B\wedge B^+)
    + \frac12c^+[c,c].
\end{equation}
The deformation \(\mu_i^M\) adds an extra quadratic term to the Maurer--Cartan action:
\begin{equation}
    S = \int_\Sigma B\wedge\left(\mathrm dA+\frac12[A,A]\right)+\frac12tB\wedge\star B
    + c(\mathrm dA^++A\wedge A^++B\wedge B^+)
    + \frac12c^+[c,c],
\end{equation}
where we may set \(t=1\) (or absorb it into \(B\)), which yields the usual action for first-order Yang--Mills theory.
\end{example}

\begin{example}[Ordinary sigma model]
Let \((\Sigma,g)\) be an \((n+1)\)-dimensional pseudo-Riemannian manifold (worldvolume), and let \(Y\) be a smooth manifold (target space). Define the shifted cotangent bundle
\begin{equation}
    X\coloneqq\mathrm T^*[n]Y\overset p\twoheadrightarrow Y
\end{equation}
with vanishing homological vector field \(Q=0\) and the canonical symplectic structure is a symplectic \(n\)-algebroid.
The corresponding AKSZ theory is given by the action
\begin{equation}
    S = \int_\Sigma A_i\wedge\mathrm d\phi^i
\end{equation}
where \(\phi\in\Omega^0(\Sigma; Y)\) and \(A\in\Omega^n(\Sigma;\phi^*\mathrm T^*Y)\) and \(\mathrm d\) is the derivative of a smooth map between manifolds so that \(\mathrm d\phi\in\Omega^1(\Sigma;\phi^*\mathrm TY)\). This is a trivial theory describing a constant \(\phi\) and closed \(A\).

Fix a pseudo-Riemannian metric \(M_{ij}\) on \(Y\). Then the Levi-Civita connection on \(\mathrm TY\) induces a canonical Ehresmann connection
\begin{equation}
    \mathrm TX
    =
    \mathrm V_p \oplus H
\end{equation}
on \(\mathrm TX\), where \(p\colon X\to Y\) is the canonical projection. Then \((p,H)\) is an admissible fibration, and a Hodge structure on it is given by the inverse Riemannian metric \(M^{ij}\). The corresponding homotopy Manin theory is
\begin{equation}
    S = \int_\Sigma A_i\wedge d\phi^i - \frac12M^{ij}A_i\wedge\star A_j.
\end{equation}
The equation of motion for \(A\) is now
\begin{equation}
    A_i = \star^{-1}M_{ij}\mathrm d\phi^j.
\end{equation}
Integrating out \(A\), we obtain
\begin{equation}
    S = \frac12\int_\Sigma M_{ij}(\star^{-1}\mathrm d\phi^i)\wedge\mathrm d\phi^j,
\end{equation}
which is the action for the ordinary sigma model on the Riemannian manifold \((X,M)\).
\end{example}

\begin{example}[Freedman--Townsend form of principal chiral model]
Let \(\mathfrak g\) be a Lie algebra. Then \((\mathfrak d,\mathfrak g^*[d-3])\) is a Manin pair, where
\begin{equation}
    \mathfrak d=\mathfrak g^*[d-3]\oplus\mathfrak g
\end{equation}
carries Lie bracket
\begin{align}
    [x,y]_{\mathfrak d}&=[x,y]_{\mathfrak g} & 
    [x,\tilde x]_{\mathfrak d}&=\operatorname{coad}_x(\tilde x)&
    [\tilde x,\tilde y]&=0
\end{align}
for \(x,y\in\mathfrak g\) and \(\tilde x,\tilde y\in\mathfrak g^*[d-3]\), where \(\operatorname{coad}\) is the coadjoint representation.

Suppose that \(\kappa\colon\mathfrak g\otimes\mathfrak g\to\mathbb R\) is an invariant metric on \(\mathfrak g\), which by musical isomorphism induces an isomorphism
\begin{align}
    \varkappa^\sharp\colon\mathfrak g&\to\mathfrak g^*\\
    x&\mapsto\kappa(x,-).
\end{align}
Then 
\begin{align}
    M\colon\mathfrak d&\to\mathfrak d \\
    (x\oplus\tilde x[d-3])&\mapsto(0\oplus\varkappa^\sharp(x))
\end{align}
(for any \(x\in\mathfrak g\), \(\tilde x\in\mathfrak g^*\)) is a Hodge structure on \((\mathfrak d,\mathfrak g^*[d-3])\).

The corresponding homotopy Manin theory on an \(n\)-dimensional (pseudo-)Riemannian manifold \(\Sigma\) has field content
\begin{multline}
    \underbrace{\underset{c^{(0)}}{\Omega^0(\Sigma;\mathfrak g^*[d-3])}}_{\mathclap{\text{\((d-4)\)\textsuperscript{th}-order ghost}}}
    \oplus
    \dotsb
    \oplus\underbrace{\underset{c^{(d-3)}}{\Omega^{d-3}(\Sigma;\mathfrak g^*[d-3])}}_{\text{ghost}}\\
    \oplus
    \underbrace{\underset{c^{(d-1)+}\eqqcolon A}{\Omega^1(\Sigma;\mathfrak g)}
    \oplus
    \underset{c^{(n-2)}\eqqcolon B}{\Omega^{n-2}(\Sigma;\mathfrak g^*[d-3])}}_{\text{fields}}
    \oplus
    \underbrace{\underset{c^{(d-1)}\eqqcolon A^+}{\Omega^{n-1}(\Sigma;\mathfrak g^*[d-3])}
    \oplus
    \underset{\mathclap{c^{(d-2)+}\eqqcolon B^+}}{\Omega^2(\Sigma;\mathfrak g)}}_{\text{antifields}}\\
    \oplus
    \underbrace{\underset{c^{(d-3)+}}{\Omega^3(\Sigma;\mathfrak g)}}_{\mathclap{\text{ghost antifield}}}
    \oplus
    \dotsb
    \oplus
    \underbrace{\underset{c^{(0)+}}{\Omega^d(\Sigma;\mathfrak g)}}_{\mathclap{\text{\((d-4)\)\textsuperscript{th}-order gh.\ antif.}}}.
\end{multline}
The corresponding action is the Freedman--Townsend formulation \cite{Freedman:1980us,Buchbinder:2024tui} of the principal chiral model:
\begin{equation}
    S = \int_\Sigma\operatorname{tr}\left(\frac12A\wedge\star A+\sum_{i=0}^{d-2}c^{(i)}\wedge\mathrm dc^{(i+1)+}
    + \sum_{i,j}c^{(i+j)}\wedge c^{+(i)}\wedge c^{+(j)}
    \right).
\end{equation}
Amongst these terms, the terms involving at most one field not of degrees \(1\) or \(2\) are
\begin{equation}
\begin{split}
    S &= \int_\Sigma\operatorname{tr}\left(B\wedge\left(\mathrm dA+\frac12[A,A]\right)+\frac12A\wedge\star A \right.\\&\left.\hspace{2cm}+ c^{(d-3)}\wedge(\mathrm d+A\wedge)B^+
   +\frac12B^+\wedge B^+\wedge c^{(d-4)}\right)+\dotsb.
\end{split}
\end{equation}
The field \(B\) enforces flatness of \(A\). We can solve this constraint as \(A=g^{-1}\mathrm dg\), so that the action becomes that of the principal chiral model.

\end{example}

\begin{example}[Broccoli--Değer--Theisen theory]
Suppose that \(\mathfrak h^0\) is a Lie algebra, and let \(V\) be a vector space. Let
\begin{equation}
    \alpha\colon(\mathfrak h^0)^{\wedge i}\to V
\end{equation}
be a \(V\)-valued Lie algebra cocycle. Then
\begin{equation}
    \mathfrak h\coloneqq \mathfrak h^0\oplus V[i-1]
\end{equation}
admits an \(L_\infty\)-algebra structure that is a central extension of \(\mathfrak h^0\). Let us then define
\begin{equation}
    \mathfrak d = \mathfrak h^0\oplus V[i-1]\oplus V^*[d-p-2]\oplus(\mathfrak h^0)^*[d-3]
\end{equation}
to be the \(L_\infty\)-algebra in which any bracket vanishes when one or more of the arguments belong to \(V^*[d-p-2]\oplus(\mathfrak h^0)^*[d-3]\). This admits a canonical cyclic structure. Then
\begin{equation}
    \mathfrak g=V^*[d-p-2]\oplus(\mathfrak h^0)^*[d-3]\subset\mathfrak d
\end{equation}
is an  Abelian \(L_\infty\)-subalgebra, which is easily seen to be admissible.
Further suppose that \(\mathfrak g\) admits an invariant inner product, and also equip \(V\) with an inner product. A Hodge structure is then given by
\begin{align}
    M\colon\mathfrak d&\to\mathfrak d\\
    (\tilde a,\tilde b,b,a)&\mapsto(0,0,\tilde b,\tilde a)
\end{align}
where we have implicitly identified \(\mathfrak g\) with \(\mathfrak g^*\) and \(V\) with \(V^*\), via the musical isomorphisms induced by the inner products on $\mathfrak g$ and $V$. The resulting action is then
\begin{equation}
S = \int \tilde A\wedge F[A] + \frac12M_{ab}A^a\wedge\star A^a + \tilde B\wedge(\mathrm dB+A\wedge\dotsb\wedge A)+\frac12M_{ij}B^i\wedge\star B^j,
\end{equation}
where the field content is
\begin{align}
A&\in\Omega^1(M;\mathfrak g)\\
B&\in\Omega^p(M;\mathfrak h)\\
\tilde B&\in\Omega^{d-p-1}(M;\mathfrak h^*)\\
\tilde A&\in\Omega^{d-2}(M;\mathfrak g^*).
\end{align}
The equations of motion are then
\begin{align}
    \mathrm dB+A\wedge\dotsb\wedge A&=0\\
    \mathrm d\tilde B_i+M_{ij}\star B^i&=0\\
    \mathrm dA+A\wedge A&=0\\
    \mathrm d\tilde A_a+M_{ab}\star A^i&=0,
\end{align}
which reproduces the equations in \cite[§3]{Broccoli:2021pvv}.
\end{example}

\subsection{Two dimensions}
A Lie 1-algebroid is the same as a Lie algebroid \(\mathfrak d\twoheadrightarrow Y\), and a symplectic Lie (1-)algebroid (the target space for a one-dimensional AKSZ sigma model) is the same as a Poisson manifold \((Y,\pi)\), or rather the associated cotangent Lie algebroid \(X\coloneqq\mathrm T^*_\pi[1]Y\), whose underlying vector bundle is the cotangent bundle \(\mathrm T^*[1]Y\) and whose anchor is given by \(\pi^\sharp\colon\mathrm T^*Y\to\mathrm TY\), and the symplectic form is the canonical pairing between \(\mathrm T^*[1]Y\) and \(\mathrm TY\).

The AKSZ sigma model in two dimensions is the Poisson sigma model \cite{Ikeda:1993fh,Schaller:1994es,Schaller:1994uj} (reviewed in \cite{Schaller:1995xk,contreras}), given by
\begin{equation}
    S = \int_\Sigma A_i\wedge\mathrm d\phi^i- \frac12\pi^{ij} A_i\wedge A_j
\end{equation}
for \(\phi\in\Omega^0(\Sigma;Y)\) and \(A\in\Omega^1(\Sigma;\phi^*\mathrm T^*Y)\).

\begin{example}
Given a Poisson manifold \((Y,\pi)\) with a Riemannian metric \(M_{ij}\) on \(Y\),
we have the admissible fibration
\begin{equation}
    \mathrm T_\pi^*[1]Y\twoheadrightarrow Y
\end{equation}
of the cotangent Lie algebroid \(X\coloneqq\mathrm T^*_\pi[1]Y\)
together with an Ehresmann connection corresponding to the Levi-Civita connection of \(M\).
A Hodge structure on this is given by the inverse Riemannian metric \(M^{ij}\).

The action of the homotopy Manin theory is
\begin{equation}
    S = \int_\Sigma A_i\wedge\mathrm d\phi^i - \frac12\pi^{ij} A_i\wedge A_j - \frac12M^{ij}A_i\wedge\star A_j
\end{equation}
This action is not quite invariant under the full Lie algebroid gauge symmetry
\begin{align}
    \delta\phi^i&=-\pi^{ij}\epsilon_j&
    \delta A_{\mu i} &= A_{\mu i}+\partial_\mu\alpha_i+\frac12\partial_i\pi^{jk}A_{\mu j}\alpha_k
\end{align}
for \(\alpha\in\Omega^0(\Sigma,\phi^*\mathrm T^*X)\) due to the mass term.

Now, we can integrate out \(A\) as
\begin{equation}
    A_i = (M^\sharp+\star \pi^\sharp)^{-1}{}_{ij}\mathrm d\phi^j,
\end{equation}
where \(M^\sharp\colon\mathrm T^*Y\to\mathrm TY\) is induced by the Riemannian metric \(M\),
so that
\begin{equation}
    S = \int_\Sigma g_{ij}\mathrm d\phi^i\wedge(M^\sharp+\star \pi^\sharp)^{-1}{}_{ij}\star\mathrm d\phi^j
\end{equation}
\end{example}

\begin{example}
Let \((Y,\pi)\) be a linear Poisson manifold, i.e.~a Lie coalgebra, and let \(X\coloneqq\mathrm T^*_\pi[1]Y\).
Then we have the admissible fibration consisting of the graded vector bundle
\begin{equation}
    \mathrm T^*[1]Y\cong Y\times Y^*[1]\twoheadrightarrow Y^*[1]
\end{equation}
equipped with the trivial Ehresmann connection.
A Hodge structure is given by a nondegenerate bilinear metric \(M\) on \(X\).
Then the action is
\begin{equation}
    S = \int_\Sigma A_i\wedge\mathrm d\phi^i-\frac12\pi^{ij} A_i \wedge A_j - \frac12M_{ij}\phi^i\wedge\star\phi^j
\end{equation}
for \(\phi\in\Omega^0(\Sigma)\otimes X\) and \(A\in\Omega^1(\Sigma)\otimes X\). We can integrate out \(\phi\) as
\begin{equation}
    \phi^i =   M^{ij}\star\mathrm dA_j
\end{equation}
so that
\begin{equation}
    S = \int_\Sigma\frac12M^{ij}\mathrm dA_i\wedge\star\mathrm dA_j-\frac12\pi^{ij}A_i\wedge A_j.
\end{equation}
This theory contains 2-dimensional Maxwell theory as the special case
where \(Y\) is a one-point space. In the general case, the theory describes a deformation of a
$\mathrm U(1)^n$ gauge theory, which is reminiscent of the Proca theory.
\end{example}

\subsubsection{Yang--Baxter sigma models}
Yang--Baxter sigma models \cite{Klimcik:2002zj,Klimcik:2008eq} (reviewed in \cite{yoshida,Hoare:2021dix}), which are integrable deformations of the principal chiral model or sigma models on symmetric spaces, may be naturally realised as homotopy Manin theories on Poisson--Lie groups \cite{Dri83a,Dri83b} (reviewed in \cite{Grabowski1995,Meusburger:2021cxe}).

Let \((G,\pi)\) be a Poisson--Lie group whose Lie algebra is \(\mathfrak g\), so that \(X\coloneqq\mathrm T^*_\pi[1]G\). The Poisson structure of \(G\) induces a Lie bialgebra structure on \(\mathfrak g\). The corresponding Poisson sigma model is
\begin{equation}
    S = \int_\Sigma \operatorname{tr}(g^{-1}\mathrm dg\wedge A) - \frac12\pi^\sharp(A\wedge A),
\end{equation}
where \(g\in\Omega^0(\Sigma;G)\) and \(A\in\Omega^1(\Sigma;\mathfrak g^*)\).
An admissible fibration on this Poisson--Lie group is given by the vector bundle \(X=\mathrm T^*_\pi[1]G\twoheadrightarrow G\) together with the canonical Ehresmann connection given by the canonical trivialisation \(\mathrm T^*[1]G\cong G\times\mathfrak g^*[1]\). A Hodge structure (metric) is given by the choice of an invariant metric on \(\mathfrak g\). Using this, we can deform the Poisson sigma model to
\begin{equation}
    S = \int_\Sigma\operatorname{tr}(g^{-1}\mathrm dg\wedge A) - \frac12\pi^\sharp(A\wedge A) - \frac12A\wedge\star A.
\end{equation}
The equation of motion for \(A\) is
\begin{equation}
    A+\star\pi^\sharp(A,-) = \star g^{-1}\mathrm dg.
\end{equation}
So, integrating \(A\) out, we obtain
\begin{equation}
    S = \int_\Sigma\operatorname{tr}\left(g^{-1}\mathrm dg(1+\star\pi^\sharp)^{-1}\wedge*g^{-1}\mathrm dg\right),
\end{equation}
which is seen to be the action for the Yang--Baxter sigma model.

\subsection{Three dimensions}
A symplectic Lie 2-algebra is equivalent to a Courant algebroid \cite{Roytenberg:1999mny,Roytenberg:2002nu}. Specifically, a Courant algebroid \((E\twoheadrightarrow X,\langle,\rangle,[,])\) corresponds to a symplectic Lie 2-algebroid
\begin{equation}
    \mathrm T^*[2]X\oplus E[1]\twoheadrightarrow X.
\end{equation}
Picking local coordinates, we have the action
\begin{equation}
    S = \int_\Sigma D\phi\wedge B + A\wedge\mathrm dA+A\wedge A\wedge A + B\wedge\rho(A).
\end{equation}
A Courant algebroid over a single point is the same as a Lie algebra \(\mathfrak d\) with an invariant metric. In this case, the corresponding AKSZ theory is (three-dimensional) Chern--Simons theory
\begin{equation}
    \int_\Sigma\operatorname{tr}\left(A\wedge\mathrm dA+\frac23A\wedge A\wedge A\right)
\end{equation}
for \(A\in\Omega^1(\Sigma;\mathfrak d)\).
A choice of a Manin pair \(\mathfrak g\subset\mathfrak d\) and a Hodge structure \(M\colon\mathfrak d/\mathfrak g\to\mathfrak g\) leads to the action
\begin{equation}
    \int_\Sigma\operatorname{tr}\left(A\wedge\mathrm dA+\frac23A\wedge A\wedge A\right)+A\wedge\star MA,
\end{equation}
which is the Manin theory \cite{Arvanitakis:2024dbu}.

\acknowledgments
A.S.A. was supported by the \foreignlanguage{dutch}{FWO-Vlaanderen} through
the project G006119N, as well as by the\foreignlanguage{dutch}{Vrije Universiteit Brussel} through
the Strategic Research Program `High-Energy Physics', and by an FWO Senior Postdoctoral Fellowship; he is currently supported by the Croatian Science Foundation
project `HigSSinGG — Higher Structures and Symmetries in Gauge and Gravity
Theories' (IP--2024--05--7921).

\bibliographystyle{unsrturl}
\bibliography{biblio}
\end{document}